%% file: main.tex
\documentclass[runningheads]{llncs}
\usepackage{graphicx}
\usepackage{cite}
\usepackage{amsmath,amssymb,amsfonts}
\usepackage{mathtools}
\usepackage{algorithmic}
\usepackage{graphicx}
\usepackage{textcomp}
\usepackage{tikz}
\usepackage{caption}
\usepackage{cuted}
\usepackage[utf8]{inputenc}
\usepackage{pgfplots} 
\usepackage{pgfgantt}
\usepackage{pdflscape}
\usepackage{appendix}
\usepackage{pst-plot}
\usetikzlibrary{spy}
\usetikzlibrary{positioning,calc}
\usetikzlibrary{decorations.pathmorphing,calc,shapes,shapes.geometric,patterns}
\usetikzlibrary{shapes.multipart}
\usepackage{xfrac}
\usepackage{colortbl}
\usepackage{cancel}
\usepackage{comment}
\usetikzlibrary{arrows,positioning,calc,intersections}
\usetikzlibrary{datavisualization.formats.functions}
\def\BibTeX{{\rm B\kern-.05em{\sc i\kern-.025em b}\kern-.08em
    T\kern-.1667em\lower.7ex\hbox{E}\kern-.125emX}}

\usepackage{pgfplots}
\usepgfplotslibrary{fillbetween}
\usetikzlibrary{arrows, decorations.markings}

\newcommand{\setl}{\ensuremath{\mathcal{L}}}

\newcommand{\setx}{\ensuremath{\mathcal{X}}}
\newcommand{\sety}{\ensuremath{\mathcal{Y}}}
\newcommand{\setz}{\ensuremath{\mathcal{Z}}}
\newcommand{\setd}{\ensuremath{\mathcal{D}}}

\newcommand{\sete}{\ensuremath{\mathcal{E}}}
\newcommand{\bs}[1]{\boldsymbol{#1}}
\newcommand{\Wg}{W_{\text{g}}}
\newcommand{\Vg}{V_{\text{g}'}}
\newcommand{\tW}{\tilde{W}_{\text{g}}}
\newcommand{\tV}{\tilde{V}_{\text{g}'}}
 
\newcommand{\up}{u^\prime}
\newcommand{\upp}{u^{\prime\prime}}
\newcommand{\Mp}{M^\prime}
\newcommand{\Mpp}{M^{\prime\prime}}
\newcommand{\dc}{\mathcal{D}}
\newcommand{\dcp}{\mathcal{D}^\prime}
\newcommand{\dcpp}{\mathcal{D}^{\prime\prime}}
\newcommand{\C}{\mathcal{C}}
\newcommand{\Cp}{\mathcal{C}^\prime}
\newcommand{\Cpp}{\mathcal{C}^{\prime\prime}}
\newcommand{\Cn}{\mathcal{C}^{N_T}}
\newcommand{\Cr}{\mathcal{C}^{N_R}}

\newcommand{\h}{\mathsf{H}}

\newcommand{\circlearrow}{}
\DeclareRobustCommand{\circlearrow}{%
  \mathrel{\vphantom{\rightarrow}\mathpalette\circle@arrow\relax}%
}
\newcommand{\circle@arrow}[2]{%
  \m@th
  \ooalign{%
    \hidewidth$#1\circ\mkern1mu$\hidewidth\cr
    $#1-$\cr}%
}

\begin{document}
\title{Secure Identification for Multi-Antenna Gaussian Channels}
%
%
\author{Wafa Labidi\inst{1,2,3}\orcidID{0000-0001-5704-1725} \and
Christian Deppe\inst{2,3}\orcidID{0000-0002-2265-4887} \and
Holger Boche\inst{1,3,4,5}\orcidID{0000-0002-8375-8946}}
\authorrunning{W. Labidi, C. Deppe, H. Boche}
%
\institute{Technical University of Munich, TUM School of Computation, Information and Technology, Munich, Germany \and
Technical University of Braunschweig, Institute for Communications Technology, Braunschweig, Germany \and 6G-life, 6G research hub, Germany
\and  {\color{black}{Munich Center for Quantum Science and Technology, Munich, Germany}} \and {\color{black}{Munich Quantum Valley, Munich, Germany}} \\
\email{wafa.labidi@tum.de, christian.deppe@tu-braunschweig.de, boche@tum.de}}
\maketitle              
\centerline{\bf In memory of Ning Cai}
\begin{abstract}
New applications in modern communications are demanding robust and ultra-reliable low-latency information exchange such as machine-to-machine and human-to-machine communications. For many of these applications, the identification approach of Ahlswede and Dueck is much more efficient than the classical message transmission scheme proposed by Shannon. Previous studies concentrate mainly on identification over discrete channels. For discrete channels, it was proved that identification is robust under channel uncertainty. Furthermore, optimal identification schemes that are secure and robust against jamming attacks have been considered. However, no results for continuous channels have yet been established. That is why we focus on the continuous case: the Gaussian channel for its known practical relevance. We deal with secure identification over Gaussian channels. Provable secure communication is of high interest for future communication systems. A key technique for implementing secure communication is the physical layer security based on information-theoretic security. We model this with the wiretap channel. In particular, we provide a suitable coding scheme for the Gaussian wiretap channel (GWC) and determine the corresponding secure identification capacity. We also consider Multiple-Input Multiple-Output (MIMO) Gaussian channels and provide an efficient signal-processing scheme. This scheme allows a separation of signal processing and Gaussian coding as in the classical case. 

\end{abstract}

\section{Introduction}
%
%
%
%
\subsection{Motivation} 
In the classical transmission scheme proposed by Shannon \cite{Shannon}, the encoder transmits a message over a channel, and at the receiver side, the decoder aims to estimate this message based on the channel observation. However, this is not the case for identification, a new
approach in communications suggested by Ahlswede and Dueck \cite{AhlDueck} in 1989. This new problem with the semantic aspect has enlarged the basis of information theory \cite{ahlswede21book, boche21new}. In the identification scheme, the encoder sends an identification message (also called identity) over the channel and the decoder is not interested in \emph{what} the received message is, but wants to know \emph{whether} a specific message, in which the receiver is interested, has been sent or not. Naturally, the sender has no knowledge of the message the receiver is interested in. One can assume the existence of a set of events (double exponential in the blocklength) $\{E_1,\ldots,E_N\}$, where any one of them may occur with a certain probability and is only known by the sender. The receiver is focused on an event $E_i$ and wants to know \emph{whether or not} $E_i$ occurred. In this perspective, the identification problem can be regarded as a testing of many hypothesis test problems occurring simultaneously. One of the main results in the theory of identification for Discrete Memoryless Channels (DMCs) is that the size of identification codes grows doubly exponentially fast with the blocklength. This is different from the classical transmission, where the number of messages that can be reliably communicated over the channel is exponential in the blocklength. It has been thought that the identification problem is just a generalization of the classical transmission. However, this conjecture has been disproved in \cite{feedback} and \cite{correlation}. Although feedback does not increase the transmission capacity, it does increase the identification capacity \cite{feedback, labidi2021identification, wiese2022identification}. Furthermore, it has been shown in \cite{correlation} that the correlation-assisted identification capacity has a completely different behavior than the task of message transmission. 

 Identification is an increasingly important area in communications, which has been extensively studied over the past decades. 
Message identification is much more efficient than the classical transmission scheme for many new applications with high reliability and latency requirements including machine-to-machine and human-to-machine systems {\color{black}{\cite{CompoundChannel,6Gprespective}}}, digital watermarking \cite{MOULINwatermarking,AhlswedeWatermarking,SteinbergWatermarking}, industry 4.0 \cite{industry4.0} and 6G communication systems \cite{6Gcomm}.
As mentioned in {\color{black}{\cite{tactileInternet,6Gprespective}}}, security and latency requirements have to be physically embedded. In this situation \cite{CompoundChannel}, the classical transmission is limited and an identification scheme is much more efficient. Furthermore, there has been an increasing interest in providing schemes for secure identification over noisy channels. For instance, Ahlswede and Zhang focused in \cite{AhlZhang} on identification over the discrete wiretap channel, which is a basic model considered in information-theoretic security. It was demonstrated that, in contrast to secure transmission, the identification capacity of the wiretap channel coincides with the capacity of the main channel; the channel from the legitimate sender to the legitimate receiver. {\color{black}{Secure identification can thus enable entirely new applications. The secure identification protocol achieves security by design. Indeed, with a secure message transmission capacity
above zero, the secure identification protocol achieves security without additional cost. This increases the potential of information-theoretic security in 6G applications and broadens the possibilities of physical layer security techniques \cite{Schaefer_Boche_Khisti_Poor_2017}.}} However, this is true only if it is possible to transmit data securely over the channel.
Previous work on identification focused on channels with finite input and output alphabets. Only a few studies \cite{HanBook,Burnashev} have explored identification for continuous alphabets. 
We are concerned with a continuous case: the Gaussian channel due to its practical relevance. Indeed, the Gaussian channel is a good model for many communication situations, e.g., satellite and deep space communication links, wired and wireless communications, etc. The transition from the discrete case to the continuous one is not obvious. Burnashev \cite{Burnashev} has shown that for the white Gaussian channel without bandwidth constraint and finite Shannon capacity, the corresponding identification capacity is infinite.\\
Secure identification for discrete alphabets has been extensively studied \cite{Boche2019SecureIU,IDforDataStorage,CompoundChannel, ahlswede2005secrecy, ahlswede2005transmission} over the recent decades due to its important potential use in many future scenarios. Indeed, for discrete channels, it was proved in \cite{Boche2019SecureIU} that secure identification is robust under channel uncertainty and against jamming attacks.{\color{black}{ Secure identification over MIMO Gaussian channels is exciting here because it enables the establishment of important MIMO applications. The results of this work demonstrate that MIMO signal processing remains straightforward, even for secure identification. The problem can be fundamentally addressed by integrating signal processing with coding. Recently, secure identification for discrete channels \cite{ImplementationSecureIDcodes} and identification for fading channels \cite{CodesIDfading} have already been implemented. This research task is very interesting, especially because it involves implementing solutions for MIMO Gaussian channels based on the results of this work. This could mark another important advancement towards implementing post-Shannon communication in 6G \cite{cabrera20216g} and enhancing trustworthiness in the 6G era \cite{6Gandtrustworthiness}.
 }}
{\color{black}{In \cite{spandri}, the authors develop a series of encryption schemes that have been recently proven to achieve semantic-secrecy capacity. They apply these schemes to a newly examined family of identification codes, demonstrating that incorporating secrecy into identification incurs virtually no additional cost. Consequently, the groundwork has been laid for the practical application of secure identification.}}
\subsection{Contributions}
Although many researchers are now addressing the problem of secure identification for continuous alphabets, no results have yet been established. This raises the question of whether secure identification can be reached for continuous channels. This still remains an open problem. We completely solve the Gaussian case by giving the secure identification capacity of the GWC. The wiretapper, in contrast to the discrete case, is not limited anymore and has an infinite alphabet. This is advantageous for the wiretapper since he has no limit on the hardware resolution. This means that the received signal can be resolved with infinite accuracy.
Existing cryptographic approaches commonly used for wireless local area networks can be broken with sufficient computing power. In contrast, information-theoretic security provides a tool for designing codes for a specific model that is proven to be unbreakable. We consider information theoretical security. In our coding scheme, the authorized sender wants to transmit a secure identification message to the authorized receiver so that the receiver is able to identify his message.
The unauthorized party is a wiretapper who can wiretap the transmission. He tries to identify an unknown message.
We have developed a coding scheme so that secure identification over a Gaussian channel is possible. The receiver can identify a message with high probability. Furthermore, the wiretapper is not able to identify a message with high probability. The secure identification capacity of the wiretap channel was only determined in the discrete case \cite{AhlZhang}. In our paper, we compute the secure identification capacity for the GWC, which is a key metric to assess the security level of an identification scheme. In our paper, we use the results about semantic security for Gaussian channels established in \cite{WieseBoche}.
\subsection{Outline}
In Section \ref{sec:preliminaries}, we introduce the channel model, the definitions, and the main results for transmission and identification over wiretap channels. In Section \ref{sec: Identification for GWC}, we provide a coding scheme for a secure identification over the GWC and prove the dichotomy theorem for this case. In Section \ref{sec: Identification for MIMO}, we use the Single-Input Single-Output (SISO) Gaussian results to complete the proofs for the MIMO case and elaborate an effective MIMO signal-processing scheme. Section \ref{sec:conclusions} contains concluding remarks and proposes potential future research in this field.

\section{Preliminaries} \label{sec:preliminaries}
In this section, we introduce the channel models and the notation that will be used throughout the paper. We also recall some basic definitions and known results about classical transmission as well as message identification over DMCs.
\subsection{Notation}
$\mathbb{R}$ denotes the sets of real numbers; $H(\cdot)$, $I(\cdot ;\cdot)$ are the entropy and mutual information, respectively; $\mathbb{H}$ denotes the binary entropy; $d(\cdot,\cdot)$ denotes the total variation distance between two probability distributions over the same set; all logarithms and information quantities are taken to the base $2$; the space of probability distribution on the finite set $\mathcal{A}$ is denoted by $\mathcal{P}(\mathcal{A})$; $|\boldsymbol{a}|$ denotes the L$_1$ norm of a vector $\boldsymbol{a}$; $\mathbf{A}^\h$ stands for the Hermitian transpose of the matrix $\mathbf{A}$.
\subsection{Definitions and Auxiliary Results}
We first start with the discrete case and we then characterize the secure identification capacity for the Gaussian channel.
\begin{definition}
A discrete memoryless channel (DMC)\index{channel! discrete memoryless (DMC)}
is a triple $(\setx,\sety,W)$, where \setx{} and 
\sety{} are finite sets denoted as input- respectively output alphabet, and 
\mbox{$W=\left\{W(y|x):x\in\setx, y\in\sety\right\}$} is a stochastic matrix.
The probability for a sequence $y^n=(y_1,\ldots,y_n) \in\setx^n$ to be received if 
$x^n=(x_1,\ldots,x_n) \in\setx^n$ was sent is defined by
$$
W^n(y^n|x^n)=\prod_{t=1}^n W(y_t|x_t),
$$
where $n$ is the number of channel uses.
\end{definition}
Ahlswede and Dueck defined in \cite{AhlDueck} identification codes for a DMC as the following.
\begin{definition}
 A deterministic $(n,N,\lambda_1,\lambda_2)$ identification code for a DMC $W$ is a family of pairs $\left\{(u_i,\setd_i),\quad i=1,\ldots,N \right\}$ 
 with 
\begin{equation*} u_i \in \setx^n,\quad \mathcal D_{i} \subset {\mathcal Y} ^{n},~\forall ~i\in \{1,\ldots,N\} \end{equation*}
 such that for $\lambda_1+\lambda_2<1$ we have:
\begin{align}
& W^n(\setd_i^c|u_i) \leq \lambda_1 \quad  \forall i, \\
& W^n(\setd_j|u_i) \leq \lambda_2 \quad \forall i\neq j.
\end{align}
\end{definition}
\begin{definition} \label{def:randomizeD id}
 A randomized $(n,M,\lambda_1,\lambda_2)$ identification code for a DMC $W$ is a family of pairs $\left\{(Q(\cdot|i),\setd_i),\quad i=1,\ldots,N \right\}$ 
 with 
 \begin{equation*} Q(\cdot|i) \in \mathcal P \left ({{\mathcal X}^{n} }\right),\quad \mathcal D_{i} \subset {\mathcal Y} ^{n},~\forall ~i\in \{1,\ldots,N\} \end{equation*}
 such that for $\lambda_1+\lambda_2<1$ we have:
\begin{align}
& \sum_{x^n \in \setx^n}Q(x^n|i) W^n(\setd_i^c|x^n) \leq \lambda_1 \quad \forall i, \\
& \sum_{x^n \in \setx^n}Q(x^n|j) W^n(\setd_i|x^n) \leq \lambda_2\quad \forall i \neq j.
\end{align}
\end{definition}

One of the main results in identification theory is the identification coding theorem proved by Ahlswede and Dueck. Contrary to transmission, where the number of messages that can be reliably communicated over the channel is exponential in the blocklength, the size of identification codes grows doubly exponentially fast with the blocklength. Han and Verd\`{u} proved the strong converse in \cite{HanVerdu}.
\begin{theorem}\cite{AhlDueck, HanVerdu}
	Let $W$ be a {\emph{finite DMC}}, $N(n,\lambda)$ the maximal number such that an $(n,N, \lambda_1, \lambda_2)$ identification code for $W$ exists with $\lambda_1, \lambda_2 \leq \lambda $, and let $C(W)$ be the Shannon capacity of $W$ then:
	\begin{equation*}
	C_{ID}(W) \triangleq \lim_{n \to \infty} \frac{1}{n} \log \log N (n, \lambda)
	= C(W)\quad \text{for }  \lambda \in (0, \frac{1}{2}),
	\end{equation*} \label{IDtheorem}
	where $C_{ID}(W)$ defines the identification capacity of the DMC $W$.
\end{theorem}
\begin{remark}
In contrast to transmission over DMCs, randomization is essential to achieve the identification capacity. 
 We can also achieve this double exponential growth of the code size by using deterministic codes and by considering rather the average error probability (see \cite{explicitConstr}). However, in practice, it is necessary to consider the maximal error for each message. In the sequel, we will stick to randomized identification codes. Contrary to Transmission codes, the decoding sets of an identification code may overlap.
\end{remark}
\begin{figure}[!t]
\centering
\input{figures/Gaussian_Channel}
\caption{Message transmission over the discrete-time Gaussian channel.}
\label{Gaussianpic}
\end{figure}
\begin{definition}
The discrete-time Gaussian channel is a triple $(\setx,\sety,\Wg)$ given by
\begin{equation}
\Wg \colon y_i=x_i+\xi_i,\quad \forall i \in \{1,\ldots,n\}, \label{GaussianChannel}
\end{equation}
where $x^n=(x_1,x_2,\ldots,x_n) \in \setx^n \subset \mathbb{R}^n$ and $y^n=(y_1,y_2,\ldots,y_n) \in \sety^n\subset \mathbb{R}^n$ are the channel input and output blocks, respectively. The sequence $\xi^n=(\xi_1,\xi_2,\ldots,\xi_n)$ denotes the channel noise. The RVs 
 $\xi_i$ are independent and identically distributed (i.i.d.) for $i \in \{1,\ldots,n\}$. Each noise sample $\xi_i,\ i=1,\ldots,n$ is drawn from a normal distribution denoted by g with variance $\sigma^2$. If the noise variance is zero or the input is unconstrained, the Shannon capacity of the channel is infinite. Message transmission over the Gaussian channel model is depicted in Fig. \ref{Gaussianpic}. The sender wants to transmit a message $L \in \{1,\ldots, M\}$ over the Gaussian channel $W_g$. The message $L$ is encoded into a codeword $x^n$ and sent over $W_g$. At the receiver's side, the decoder attempts to estimate the sent message, denoted by $\hat{L}$, based on the channel output $y^n$.
The most common limitation on the input is a power constraint, e.g.,
   \begin{equation} \frac{1}{n} \sum_{i=1}^{n} x_i^2 \leq  P .\label{PowerConst} \end{equation}
 \end{definition}
In the sequel, we consider the average power constraint described in \eqref{PowerConst}. We define the constrained input set $\setx_{n,P}$ as follows:  
 \begin{equation}
     \setx_{n,P}=\left\{ x^n \in \setx^n \subset \mathbb{R}^n \colon \sum_{i=1}^{n} x_i^2 \leq n \cdot P \right\}. \label{inputConstrained}
 \end{equation}
 We denote the Shannon capacity and the identification capacity of $\Wg$ by $C(\text{g},P)$ and $C_{ID}(\text{g},P)$, respectively.\\
We extend the definitions of a transmission code and an identification code to the Gaussian channel $\Wg$.
 \begin{definition}
 A randomized $(n,M,\lambda)$ transmission code for the Gaussian channel $\Wg$ is a family of pairs $\{(Q(\cdot|i),\setd_i),\ i=1,\dots,M\}$ with
\begin{equation*} Q(\cdot|i) \in \mathcal P \left ({{\mathcal X}_{n,P} }\right), \quad \mathcal D_{i} \subset {\mathcal Y} ^{n} \end{equation*} such that for all $i\neq j$
\begin{equation} \mathcal D_{i}\cap \mathcal D _{j}=\emptyset \end{equation}, 
for all $i$ and $ \lambda \in (0,1)$
\begin{equation}\int _{x^{n}\in {\mathcal X} _{n,P}} Q(x^n|i) \Wg^{n}(\mathcal D_{i}|x^{n}) d^n x^n\geq 1-\lambda. \end{equation} 
 \end{definition}
\begin{definition} A randomized $(n,N,\lambda_1,\lambda_2)$ identification code  for the Gaussian channel $\Wg$ is a family of pairs $\left\{(Q(\cdot|i),\setd_i),\quad i=1,\ldots,N \right\}$ with 
\begin{equation*} Q(\cdot|i) \in \mathcal P \left ({{\mathcal X}_{n,P} }\right),\quad \mathcal D_{i} \subset {\mathcal Y} ^{n},~\forall ~i\in \{1,\ldots,N\} \end{equation*}
such that for all $i,j \in \{1,\ldots,N\}$, $i \neq j$ and $\lambda_1+\lambda_2 <1$
\begin{align}
& \int_{x^n \in \setx_{n,P}} Q(x^n|i) \Wg^n(\setd_i^c|x^n) d^n x^n \leq\lambda_1,  \\
& \int_{x^n \in \setx_{n,P}} Q(x^n|j) \Wg^n(\setd_i|x^n) d^n x^n \leq\lambda_2 \quad \forall i\neq j.
\end{align} \end{definition}
The identification coding theorem is extended to the Gaussian case. It is worth pointing out that the direct part is similar in spirit to the discrete case. For more details, we refer the reader to \cite{MasterThesis}. The converse part is established in \cite{Burnashev}.
\begin{theorem} \cite{Burnashev} \label{theorem:ID_Gaussian}
Let $N(n,\lambda)$ be the maximal number such that an $(n,N,\lambda_1,\lambda_2)$ identification code for the channel $\Wg$ exists with $\lambda_1, \lambda_2 \leq \lambda$ and $\lambda \in (0,\frac{1}{2})$ then:
\begin{equation*}
C_{ID}(\text{g},P) \triangleq \lim_{n \to \infty} \frac{1}{n} \log \log N (n, \lambda)
= C(\text{g},P).
\end{equation*}
\end{theorem}

\begin{definition}
A discrete memoryless wiretap channel is a quintuple \\$( \setx,\sety,\setz,W,V) $, where $\setx$ is the finite input alphabet, $\sety$ is the finite output alphabet for the legitimate receiver, $\setz$ is the finite output alphabet for the wiretapper, $W=\{W(y|x):x \in \setx, y\in \sety\}$ is the set of the transmission matrices whose output is available to the legitimate receiver, and $V=\{V(z|x):x\in \setx,z\in \setz \}$ is the set of transmission matrices whose output is available to the wiretapper. For notational simplicity, we denote the wiretap channel with the pair $(W,V)$.
\end{definition}
The discrete wiretap channel was first introduced by Wyner \cite{Wyner}. He determined the secrecy capacity when the channel is physically degraded, i.e., when $X-Y-Z$ is a Markov chain and under weak secrecy criterion. Wyner introduced in \cite{Wyner}  transmission wiretap codes, which are transmission codes fulfilling a secrecy requirement. Csiszár and Körner \cite{BC} extended the results for the non-degraded case depicted in Fig. \ref{fig:wiretap}. The channel has two outputs: One is for the legitimate receiver (Bob) and the other is for the eavesdropper (Eve). The sender (Alice) has to send a message $L$ to the legitimate receiver, who then outputs an estimation $\hat{L}$. The message $L$ should be kept secret from the eavesdropper. Thus, we have a secrecy requirement (in this paper, we restrict our attention to strong secrecy).
 \begin{equation} I(L;Z^n) \leq \delta_1\quad \text{with  }\delta_1>0,  \label{SecrecyReq} \end{equation} where $Z^n$ is the output to the eavesdropper. The legitimate receiver should be able to decode his message with a small error probability $P_e^{(n)}$.
 \begin{equation} P_e^{(n)} \triangleq \Pr[\hat{L}\neq L] \leq \delta_2\quad \text{with  } \delta_2>0. \label{reliabReq} \end{equation}
In the following, we recall the definition of a randomized transmission code for the discrete wiretap channel.
 \begin{definition}
An $(n,M,\lambda)$ randomized transmission code for the wiretap channel $(W,V)$ is a family of pairs $\{(Q(\cdot|i),\setd_i),\ i=1,\dots,M\}$ with 
\begin{align*} Q(\cdot|i)\in \mathcal P \left ( \setx^{n} \right),\mathcal D_{i} \subset {\mathcal Y} ^{n}, && \forall ~i\in \{1,\ldots,N\}  \end{align*} 
such that for all $i\in \{1,\ldots,M\}$ and $i\neq j$
\begin{align}
\sum_{x^n \in \setx^n} Q(x^n|i) W^n(\setd_i^c|x^n)  &\leq\lambda,  \\
 \setd_i \cap \setd_j & =  \emptyset,\\
I(L;Z^n) & \leq \lambda, \label{stronSecrecy}
\end{align}
where $L$ is a uniformly distributed random variable (RV) on $\{1,\ldots,M\}$ as defined earlier. The RV $Z$ denotes the output of Eve's channel $V$ i.e., the wiretapper's observation.
\end{definition} 
\begin{figure}[!t]
\centering
\input{figures/wiretap.tex}
\caption{The wiretap channel model.}
\label{fig:wiretap}
\end{figure}
Ahlswede and Zhang defined in \cite{AhlZhang} identification codes for discrete wiretap channels as follows.
\begin{definition}
 A randomized $(n,N,\lambda_1,\lambda_2)$ identification code for a wiretap channel $(W,V)$ is a family of pairs  $\{(Q(\cdot|i),\setd_i),\ i=1,\dots,N\}$  with 
\begin{equation*} Q(\cdot|i) \in \mathcal P \left ({{\mathcal X}^{n} }\right),\quad \mathcal D_{i} \subset {\mathcal Y} ^{n},~\forall ~i\in \{1,\ldots,N\} \end{equation*} such that for all $i,j \in \{1,\ldots,N\}$, $i \neq j$ and any $\sete \in \setz^n$
\begin{align} \sum _{x^{n}\in {\mathcal X} ^{n}}{Q(x^n|i) W^{n}\left ({\mathcal D_{i}|x^{n}}\right)} & \ge 1-\lambda _{1}, \\\sum _{x^{n}\in {\mathcal X} ^{n}}{Q(x^n|j) W^{n}\left ({\mathcal D_{i}|x^{n}}\right)} & \le \lambda _{2}, \\\sum _{x^{n}\in {\mathcal X} ^{n}}Q(x^n|j) V^{n}(\mathcal E|x^{n})\notag \\ \qquad \,\,\,+\ \sum _{x^{n}\in {\mathcal X} ^{n}}Q(x^n|i)V^{n}(\mathcal E^{c}|x^{n}) &\geq 1-\lambda. \label{Wcond} \end{align}
We denote with $N_S(n,\lambda)$ the maximal cardinality such that an $(n,N,\lambda_1,\lambda_2)$ identification wiretap code for the channel $(W,V)$ exists. 
\end{definition}
It was shown in \cite{Igor} that \eqref{Wcond} implies strong secrecy. 
\begin{definition}
    \begin{itemize}
\item The secure identification rate $R$ of a wiretap channel is said to be achievable if for all $ \lambda \in (0,1)$ there exists a $n(\lambda)$, such that for all $n \geq n(\lambda) $ there exists an $(n,N,\lambda,\lambda)$ wiretap identification code. 
\item The secure identification capacity $C_{SID}(W,V)$ of a wiretap channel $(W,V)$ is the supremum of all achievable rates.
   \end{itemize}
\end{definition}
Ahlswede and Zhang proved in \cite{AhlZhang} the following \emph{dichotomy} theorem.
\begin{theorem} \cite{AhlZhang} 
Let $C(W)$ be the Shannon capacity of the channel $W$ and let $C_S(W,V)$ be the secrecy capacity of the wiretap channel $(W,V)$, then the secure identification capacity of $(W,V)$ is computed as follows:
\begin{equation*}
C_{SID}(W,V)= \begin{cases} C(W) & \text{if } C_S(W,V) > 0 \\
0 & \text{if }C_S(W,V)=0. \end{cases}
\end{equation*}
\end{theorem}
\subsection{Channel Model}
We consider the following standard GWC model:
\begin{equation}
\begin{aligned}
\Wg\colon y_i =x_i+\xi_i, &&  \forall i \in \{1,\ldots,n\}, \\
\Vg \colon z_i =x_i + \phi_i, && \forall i \in \{1,\ldots,n\},
 \end{aligned} \label{channelModel}
 \end{equation}
 where $x^n=(x_1,x_2,\ldots,x_n)$ is the channel input sequence. $y^n=(y_1,y_2,\ldots,y_n)$ and $z^n=(z_1,z_2,\ldots,z_n)$ are Bob's and Eve's observations, respectively. $\xi^n=(\xi_1,\xi_2,\ldots,\xi_n)$ and $\phi^n=(\phi_1,\phi_2,\ldots,\phi_n)$ are the noise sequences of the main channel and the wiretapper's channel, respectively.
 $\xi_i$ are i.i.d and drawn from a normal distribution denoted by $\text{g}$ with zero-mean and variance $\sigma^2$. $\phi_i$ are i.i.d and drawn from a normal distribution denoted by $\text{g}'$ with zero-mean and variance $\sigma'^2$. The channel input fulfills the following power constraint. \begin{equation} \frac{1}{n} \sum_{i=1}^{n} x_i^2 \leq  P. \label{powerConstr}
 \end{equation}The input set is $\setx_{n,P}$, defined in \eqref{inputConstrained}. 
The output sets are infinite $\sety=\setz= \mathbb{R}$. We denote the GWC by the pair $(W_{\text{g}},V_{\text{g}'})$, where $\Wg$ and $\Vg$ define the Gaussian channels to the legitimate receiver and the wiretapper, respectively.
 The strong secrecy capacity of the GWC is denoted by $C(\text{g},\text{g}',P)$. We mean with $C(\text{g},\text{g}',P)$ the Shannon capacity of the GWC when strong secrecy requirement \eqref{SecrecyReq} is fulfilled.
\begin{theorem} \cite{40,WieseBoche}
Let $\Wg \colon \mathbb{R} \to \mathbb{R}$ and $\Vg \colon \mathbb{R} \to \mathbb{R}$ be Gaussian channels with noise variances $\sigma^2$ and $\sigma'^2$, respectively. The strong secrecy capacity of the GWC $(\Wg,\Vg)$ with input power constraint $P$ is given by

 \begin{equation*}
 C_S(\text{g},\text{g}',P)=\begin{cases} \frac{1}{2} \log \left( \frac{1+\frac{P}{\sigma^2}}{ 1+\frac{P}{\sigma'^2}} \right)  & \text{if}\ \sigma^2 \geq \sigma'^2 \\
 0 & \text{else}. \end{cases}
 \end{equation*}
\end{theorem}
The definition of identification codes for the single-user MIMO channel is similar to the SISO case, except for the dimension of input and output sets.
Indeed, in each time $i \in \{1,\ldots,n\}$, we send $N_T$ scalar signals and receive $N_R$ signals. Thus, compared to the SISO case, the input and output sets contain matrices instead of vectors. 

\section{Identification for Gaussian Wiretap Channels} \label{sec: Identification for GWC}
In this section, we present and prove the main result of this paper. We first provide identification codes for the GWC and determine the corresponding secure identification capacity.
\begin{theorem} \label{main Theorem}
  Let $C_{SID}(\text{g},\text{g}^\prime,P)$ be the secure identification capacity of the wiretap channel $(W,V,\text{g},\text{g}^\prime,P)$ and $C(g,P)$ the identification capacity of the main Gaussian $W(\text{g},P)$ then:
  \begin{equation*}
  C_{SID}(\text{g},\text{g}^\prime,P)= \begin{cases} C(g,P) & \text{if } C_S(\text{g},\text{g}^\prime,P) > 0 \\
  0 & \text{if }C_S(\text{g},\text{g}^\prime,P)=0. \end{cases}
  \end{equation*}
 \end{theorem}
As we mentioned, until now, it is not known whether secure identification can be reached for continuous channels. Here, we completely solve this problem for the GWC by proving Theorem \ref{main Theorem}. This theorem states that if the secrecy capacity of the GWC is positive, then the corresponding secure identification capacity coincides with the Shannon capacity of the main channel, i.e., the channel from the legitimate sender to the legitimate receiver. Thus, in this case, the channel to the wiretapper has no further influence on the secure identification capacity. For discrete wiretap channels, the alphabet of the wiretapper is discrete, i.e., the wiretapper is far more limited compared to the continuous case. It is advantageous for the wiretapper to deal with an infinite alphabet because he has no limit on the hardware resolution. This means that the received signal can be resolved with infinite accuracy. We assume that the eavesdropper has unlimited computational ability and no restrictions regarding quantization and digital hardware platform.
\subsection{Identification Codes for the GWC}
Based on the definitions of identification codes for the discrete wiretap channel in \cite{AhlZhang}, we introduce identification codes for the GWC.
\begin{definition} \label{IDwiretapCode}
 A randomized $(n,N,\lambda_1,\lambda_2)$ identification code for the GWC $(\Wg,\Vg)$ is a family of pairs $\{(Q(\cdot|i),\setd_i),\ i=1,\dots,N\}$ with \begin{align*} Q(\cdot|i) \in \mathcal P \left (\setx_{n,P}\right),\mathcal D_{i} \subset {\mathcal Y} ^{n}, && \forall ~i\in \{1,\ldots,N\} \end{align*}
such that for all $i,j \in \{1,\ldots,N\}$, $i \neq j$ and any $\sete \in \setz^n$
\begin{align}
\int_{x^n \in \setx_{n,P}} Q(x^n|i) \Wg^n(\setd_i^c|x^n) d^n x^n &  \leq\lambda_1,  \label{firstRequirement} \\
 \int_{x^n \in \setx_{n,P}} Q(x^n|j) \Wg^n(\setd_i|x^n) d^n x^n &\leq\lambda_2, \label{secondRequirement}\\
\int_{x^n \in \setx_{n,P}}  Q(x^n|j) \Vg^n(\sete|x^n) d^n x^n \notag \\ \qquad \,\,\,+\  \int_{x^n \in \setx_{n,P}} Q(x^n|i) \Vg^n(\sete^c|x^n) d^n x^n & \geq 1-\lambda.
\end{align}
\end{definition}
 
\subsection{Optimal Secure Coding Scheme for Identification}
\begin{figure}[!t]
\centering \input{figures/CodeConstrWiretap}
\caption{Identification wiretap code construction flowchart.}
\label{CodeConstruction}
\end{figure}
The idea of the direct proof is to concatenate two fundamental codes. We consider a transmission code $\Cp$ and a wiretap code $\Cpp$  as depicted in Fig. \ref{CodeConstruction}. 
For the message set $\{1,\ldots,\Mp\}$ one uses $\{1,\ldots,\Mpp\}$ as a suitable indexed set of colorings of the messages with a smaller number of colors. Both of the coloring and color sets are known to the sender and the receiver(s). Every coloring function, denoted by $T_i \colon \Cp \longrightarrow \Cpp$, corresponds to an identification message $i$. The sender chooses a coloring number $j$ randomly from the set $\{1,\ldots,\Mp\}$ and calculates the color of the identification message $i$ under coloring number $j$ using $T_i$, denoted by $T_i(j)$. We send both of $j$ with the code $\Cp$ and $T_i(j)$ with the code $\Cpp$ over the GWC. 
The receiver, interested in the identification message $i$, calculates the color of $j$ under $T_i$
and checks whether it is equal to the received color or not. In the first case, he decides that the identification message is $i$, otherwise he says it was not $i$.
For notational simplicity, we set $ \lceil\sqrt{n} \rceil=q$. The resulting $(m,N,\lambda_1,\lambda_2)$ identification code $\C=\{(Q(\cdot|i),\setd_i),\quad Q(\cdot|i) \in \mathcal{P}(\setx_{m,P}),\ \setd_i \subset \sety^m,\ i\in \{1,\ldots,N\} \}$ has blocklength $m=n+q$. The input set $\setx_{m,P}$ is defined analogously to \eqref{inputConstrained},  $$\setx_{m,P}=\left\{ x^m \in \setx^m \subset \mathbb{R}^m \colon \sum_{i=1}^{m} x_i^2 \leq m \cdot P \right\}.$$ 
 We want to show that
\begin{equation} {C_{SID}(\text{g},\text{g}',P)\geq C(\text{g},P)\quad \text{if } C_S(\text{g},\text{g}',P) > 0}. \end{equation}
 As a wiretap code for $(\Wg,\Vg)$ is also a transmission code, the bounds for the errors of the first kind and of the second kind can be computed similarly to the Gaussian case (see \cite{MasterThesis}). It is obvious that the concatenated code fulfills the power constraint. Now, it remains to show that \eqref{Wcond} holds. For this purpose, it is sufficient to prove that the wiretapper can not identify the color.
We apply the results in \cite{AhlZhang} to the Gaussian case.
If $Q_i$ is an input distribution then the generated output measure over the channel $\Vg$ denoted by $Q_i\Vg^{q}(\sete)$ is defined as the following:
\begin{align}
Q_i\Vg^{q}(\sete)& \triangleq \sum_{x^{q} \in \setx^{q}} Q(x^{q}|i) \Vg^{q}(\sete|x^{q}).
\end{align}
We use a stochastic encoder. For any region $\sete \in \setz^{q}$ and any $i\neq j$, the total variation distance between two output measures $Q_iV^{q}(\sete)$ and $Q_jV^{q}(\sete)$ is upper-bounded. Indeed, as shown in \cite{AhlZhang}, from the transmission wiretap code we can construct a transmission code satisfying the following inequality:
\begin{equation}
d\left(Q_i\Vg^{q}(\sete), Q_j\Vg^{q}(\sete)\right) \leq \epsilon. \label{4.19}
\end{equation} 
In case of the GWC, we define $Q_i\Vg^{q}(\sete)$ as follows:
\begin{equation}
Q_i\Vg^{q}(\sete)\triangleq \int_{x^{q} \in \setx^{q}} Q(x^{q}|i) \Vg^{q}(\sete|x^{q}) d^qx^q.
\end{equation}
In spirit of Burnashev's approximation \cite{Bur00}, we approximate the GWC $(\Wg,\Vg)$ by the discrete GWC $(\tW,\tV)$, where the Gaussian channels $\Wg$ and $\Vg$ are approximated by the discrete channels $\tilde{\Wg} \colon \mathcal{L}_x \to \setl_y $ and $\tilde{\Vg} \colon \mathcal{L}_x \to \setl_z$, respectively. The sets $\mathcal{L}_x $, $\setl_y$ and $\setl_z$ are finite.
 A detailed description of the quantization of the GWC is presented in Appendix \ref{QuantizationWiretapAppendix}.
 
 If $Q_i\Vg^{q}(\cdot)$ is an output measure on $\Vg$, we denote its approximation on $\tV$ by $\tilde{Q}_i\tV^{q}(\cdot)$. It was shown in \cite{Bur00} that: 
\begin{equation}
d\left(Q_i\Vg^{q}(\cdot),\tilde{Q}_i\tV^{q}(\cdot)\right) \leq \delta,\quad \delta>0. \label{approximation}
\end{equation}
We now use the stochastic encoder described in \cite{CompoundChannel} for the new quantized GWC $(\tW,\tV)$. We obtain for any region $\sete \in \setz^{q}$:
\begin{equation}
d\left(\tilde{Q}_i\tV^{q}(\sete),\tilde{Q}_j\tV^{q}(\sete)\right) \leq \epsilon. \label{quantizedBound}
\end{equation}
It follows from \eqref{approximation} and \eqref{quantizedBound} that:
\begin{equation}
d\left({Q}_i\Vg^{m}(\sete),{Q}_j\tV^{m}(\sete)\right) \leq \epsilon+2\delta. \label{newbound}
\end{equation}
We choose $\delta$ small enough such that $\epsilon+2\delta = \epsilon_1$ is very small. It is clear that \eqref{newbound} implies the last secrecy requirement of an identification code for the GWC $(\Wg,\Vg)$. This completes the direct proof. \hfill\ensuremath{\square}
\subsection{Characterization of Optimal Rate}
We need the following lemmas to prove the main results of the converse part.
\begin{lemma}\label{zlemma}
 	Let $Q_{z,i}$ and $Q_{z,j}$ be two distributions on $\setz^m$ and for any Lebesgue measurable $\sete \subset \setz^m$,
 	\begin{equation}
 	Q_{z,i}(\sete)+Q_{z,j}(\sete^c) > 1- \epsilon, \quad 0<\epsilon<\frac{2}{x}
 	\end{equation}
 	and let $U$ be a binary RV with uniform distribution on $\{i,j\}$ and $V^m(z^m|U=i)=Q_{z,i}$, $V^m(z^m|U=j)=Q_{z,j}$ then
 	\begin{equation}
 	I(U;Z^m) \leq \inf_{x \in (0,\frac{2}{\epsilon})} \frac{2}{x}+\log \frac{1}{1-\frac{1}{2}x\epsilon} .
 	\end{equation}
 \end{lemma}
 \begin{lemma} \label{ylemma}
 	Let $Q_{y,i}$ and $Q_{y,j}$ be two distributions on $\sety^m$ for which there exists a $\setd \subset \sety^m$ such that
 	\begin{equation}
 	Q_{y,i}(\setd) + Q_{y,j}(\setd^c) < \epsilon,
 	\end{equation}
 	and let $U$ be a binary RV with uniform distribution and  $W^m(y^m|U=i)=Q_{y,i}$, $W^m(y^m|U=j)=Q_{y,j}$ then
 	\begin{equation}
 	I(U;{Y}^m) \geq \mathbb{H}\left(\frac{1}{2}(1-\epsilon) \right).
 	\end{equation}

 \end{lemma}
 Lemma \ref{zlemma} and Lemma \ref{ylemma} were proved in \cite{AhlZhang}. 
 To prove the converse part, we have to show the two following statements.
\begin{equation}
\begin{cases}
C_{SID}(\text{g},\text{g}',P) \leq C(g,P) & \text{if } C_S(\text{g},\text{g}^\prime,P) > 0, \\
 C_{SID}(\text{g},\text{g}',P)=0 & \text{if }C_S(\text{g},\text{g}^\prime,P)=0.
\end{cases}
\end{equation}
It is obvious that the secure identification capacity can not exceed the identification capacity of the channel, i.e.,
\begin{equation} C_{SID}(\text{g},\text{g}^\prime,P) \leq C_{ID}(\text{g},P)=C(\text{g},P).\end{equation}
It remains to show that
\begin{equation}
C_S(\text{g},\text{g}^\prime,P)=0  \Longrightarrow C_{SID}(\text{g},\text{g}^\prime,P) =0. 
\end{equation}
For convenience, we show the following equivalent statement
\begin{equation}
 \left(C_{SID}(\text{g},\text{g}^\prime,P) >0 \right)  \Longrightarrow \left(C_S(\text{g},\text{g}^\prime,P)>0 \right).
\end{equation}
 We assume that the secure identification capacity is positive. This implies the existence of a wiretap identification code $(m,N,\lambda,\lambda)$ with positive rate for $(\Wg,\Vg)$ (see \emph{Definition \ref{IDwiretapCode}}). For convenience, we chose $\lambda_1=\lambda_2=\lambda$ for $\lambda < \frac{1}{2}$.
 From \eqref{firstRequirement} and \eqref{secondRequirement}, we obtain
\begin{align}
\int_{x^m \in \setx_{m,P}} Q(x^m|i) \Wg^m(\setd_i^c|x^m) d^m x^m  \nonumber \\ + \int_{x^m \in \setx_{m,P}} Q(x^m|j) \Wg^m(\setd_i|x^m) d^m x^m \leq 2\lambda. \label{TwofirstProp}
\end{align}
Let $U$ be a binary RV with uniform distribution on the set $\{i,j\}$. We denote the probability $\Wg^m(y^m|U=i)$ by $Q_{y,i}$ and $\Wg^m(y^m|U=j)$ by $Q_{y,j}$, i.e.,
\eqref{TwofirstProp} can be rewritten as the following:
\begin{equation}
Q_{y,i}(\setd^c) + Q_{y,j}(\setd) \leq 2\lambda. \label{Twop}
\end{equation}
We now compute $I(U;Y^m)$ as follows:
{\small{\begin{align}
&I(U;Y^m)  \\
&= \int_{\substack{y^m \in \sety^m,\\ u\in\{i,j\} }} P_{UY^m}(u,y^m) \log \frac{P_{UY^m}(u,y^m)}{P_U(u)\cdot P_{Y^m}(y^m)} \ du \ d^my^m \\
& = \int_{\substack{y^m \in \sety^m,\\ u\in\{i,j\} }} \Wg^m(y^m|u) P_U(u) \log \frac{\Wg^m(y^m|u)}{P_{Y^m}(y^m)} \ du \ d^my^m \\
&\overset{(a)}{=} \frac{1}{2} \int_{\substack{y^m \in \sety^m}}  \Wg^m(y^m|i) \log \frac{\Wg^m(y^m|i)}{\frac{1}{2}(\Wg^m(y^m|i)+\Wg^m(y^m|j))}  \\
&+  \Wg^m(y^m|j) \log \frac{\Wg^m(y^m|j)}{\frac{1}{2}(\Wg^m(y^m|i)+W^m(y^m|j))} \ d^my^m\nonumber \\
&\overset{(b)}{=}\int_{\substack{y^m \in \sety^m}} \frac{1}{2} Q_{i,y}(y^m) \log \frac{2Q_{i,y}(y^m)}{(Q_{i,y}(y^m)+Q_{j,y}(y^m))} \  d^my^m \\
& + \int_{\substack{y^m \in \sety^m}} \frac{1}{2} Q_{j,y}(y^m) \log \frac{2Q_{j,y}(y^m)}{(Q_{i,y}(y^m)+Q_{j,y}(y^m))} \  d^my^m. \nonumber
\end{align}}}
$(a)$ follows by the law of total probability. $(b)$ follows by the definition of $Q_{i,y}$. \\
We want to establish a lower bound for $I(U;Y^m)$. First, let $\tilde{Y}$ on $\setl_y$ be the quantized version of $Y$. That means, if $\pi(x^n) \in \mathcal{P}(\setx_{n,P})$ is an input distribution generating the output measures $Q_{i,y}(\cdot)$ and $\tilde{Q}_{i,y}(\cdot)$ on the channels $\Wg$ and $\tilde{\Wg}$, respectively we have then
\begin{equation}
d\left(Q_{y,i}(\cdot),\tilde{Q}_{y,i}(\cdot)\right) \leq \delta^\prime, \quad \delta^\prime>0. \label{newQ}
\end{equation}
From \eqref{newQ} and \eqref{Twop}, we obtain
\begin{equation}
\tilde{Q}_{y,i}(\setd^c) + \tilde{Q}_{y,j}(\setd) \leq 2\lambda+2\delta^\prime,\quad \delta^\prime>0.
\end{equation}

We then compute $I(U;\tilde{Y}^m)$ as follows:
{{\begin{align}
I(U;\tilde{Y}^m) &= \sum_{\substack{y^m \in \mathcal{L}_y^m}} \frac{1}{2} \tilde{Q}_{i,y}(y^m) \log \frac{2 \tilde{Q}_{i,y}(y^m)}{(\tilde{Q}_{i,y}(y^m)+\tilde{Q}_{j,y}(y^m))}  \\
& + \sum_{\substack{y^m \in \mathcal{L}_y^m}} \frac{1}{2} \tilde{Q}_{j,y}(y^m) \log \frac{2\tilde{Q}_{j,y}(y^m)}{(\tilde{Q}_{i,y}(y^m)+\tilde{Q}_{j,y}(y^m))}.  \nonumber
\end{align}}}
By applying Lemma \ref{ylemma} on $\tilde{Q}_{y,i}(\setd^c)$ and $\tilde{Q}_{y,j}(\setd)$, we obtain 
\begin{equation}
I(U;\tilde{Y}^m) \geq \mathbb{H}\left(\frac{1}{2}(1-2\lambda-2\delta^\prime) \right).
\end{equation}
Therefore, it follows by data processing inequality
\begin{equation}
I(U;Y^m) \geq I(U;\tilde{Y}^m) \geq \mathbb{H}\left(\frac{1}{2}(1-2\lambda-2\delta^\prime) \right). \label{boxed1}
\end{equation}
 For the same RV $U$, we define $Q_{z,i}$ as the probability $\Vg^m(z^m|U=i)$ and $Q_{z,j}$ as the probability $\Vg^m(z^m|U=j)$. Then, in terms of $Q_{z,j}$ and $Q_{z,i}$, the last condition \eqref{Wcond} of the identification wiretap code can be rewritten as follows.
\begin{equation}
Q_{z,j}(\sete)+Q_{z,i}(\sete^c) \geq 1-\lambda,\quad i\neq j,\ \forall \sete \subset \setz^m.
\end{equation}
Lemma \ref{zlemma} implies that
\begin{equation}
I(U;Z^m) \leq \inf_{x >0} \left(\frac{2}{x}+\log \frac{1}{1-\frac{1}{2}x\lambda} \right).\label{boxed2}
\end{equation}
 For $\lambda$ and $\delta$ small enough, it follows from \eqref{boxed1} and \eqref{boxed2} that there exists an index $n_0$ and a RV $U$ with $U \longrightarrow X^{n_0}\longrightarrow Y^{n_0}Z^{n_0}$ such that
\begin{equation} I(U;Y^{n_0}) > I(U;Z^{n_0 }). \end{equation}
That means that we have a code with a positive secrecy rate i.e.,
 $$C_{SID}(\text{g},\text{g}^\prime,P) >0   \Longrightarrow C_S(\text{g},\text{g}^\prime,P)>0. $$ This completes the converse part.\hfill\ensuremath{\square} 
\section{Identification for the Gaussian MIMO Channel}
\label{sec: Identification for MIMO}
In a point-to-point or single-user MIMO communication system, the transmitted signal is a joint transformation of data signals from multiple transmit antennas. The entries of the output signal, collected by multiple receive antennas, are jointly processed. Single-user MIMO communication systems, compared to Single-Input Single-Output (SISO) systems, offer higher rates, more reliability and resistance to interference. In this section, we focus on identification over the Gaussian single-user MIMO channel and compute the corresponding identification capacity. 
\subsection{Channel Model}
The channel model is depicted in Fig. \ref{MIMOModel}.
\begin{figure}[hbt!]
	\centering 
	\input{figures/MIMOModel.tex}
	\caption{The Gaussian single-user MIMO channel model.}
	\label{MIMOModel}
\end{figure}
We consider the following channel model with $N_T$ transmit antennas and  $N_R$ receive antennas:
	\begin{equation} \bs{y}_i=\mathbf{H}\bs{x}_i+\bs{\xi}_i,\quad \forall i=1,\ldots,n,
	\label{Model:MIMO}\end{equation}
where $n$, as previously mentioned, is the number of channel uses. For simplicity, we drop the index $i$.	The input vector $\boldsymbol{x} \in \Cn $ contains the $N_T$ scalar transmitted signals and fulfills the following power constraint:
\begin{equation}
    \mathbb{E}[\bs{x}^\h \bs{x}] \leq P.
\end{equation}
The output vector $\bs{y} \in \Cr$
comprises the scalar received signals of the $N_R$ channel outputs. The channel matrix
	\begin{equation*} \mathbf{H}= \begin{pmatrix} h_{11} & \ldots& h_{1N_T} \\ 
	\vdots & \ddots & \vdots \\  h_{N_R1} & \ldots & h_{N_RN_T}
	\end{pmatrix}\in \mathbb{C}^{N_R \times N_T} \end{equation*}
	is a full-rank matrix.
	The entry $h_{ij}$ represents the channel gain from transmit antenna $j$ to receive antenna $i$. 
The vector $\bs{\xi} \in \Cr$ is circularly symmetric Gaussian noise, $\bs{\xi} \sim \mathcal{N}_{\mathbb{C}}(\bs{0}_{N_T},\sigma^2 \mathbf{I}_{N_R})$.
If $\mathbf{H}$ is deterministic and perfectly known at the transmitter and the receiver, then the capacity $C(P,N_T\times N_R)$ of the MIMO channel described in \eqref{Model:MIMO} is given by \cite{Tse}:
	\begin{equation*}
	C(P,N_T\times N_R)= \max_{\mathbf{Q}:\ \substack{tr(\mathbf{Q})=P\\ Q \succeq 0}} \log \det \left(\mathbf{I}_{N_R}+\frac{1}{\sigma^2}\mathbf{H} \mathbf{Q} \mathbf{H}^\h \right). \label{eq:capacity_MIMO}
	\end{equation*}
	where $\mathbf{Q} \in \mathbb{C}^{N_T\times N_T}$ is the covariance matrix of the input vector $\bs{x}$.
	\subsection{Shannon Capacity of the Gaussian MIMO Channel}
	The capacity $C(P,N_T\times N_R)$ can be computed by converting the MIMO channel into parallel, independent and scalar Gaussian sub-channels. This conversion is based on the following singular value decomposition (SVD) of the channel matrix $\mathbf{H}$:
	\begin{equation}
	   \mathbf{H}=\mathbf{U}\mathbf{\Lambda}\mathbf{V}^\h,
	\end{equation}
where $\mathbf{U} \in \mathbb{C}^{N_R\times N_R}$ and $\mathbf{V} \in \mathbb{C}^{N_T\times N_T}$ are unitary matrices. $\mathbf{\Lambda} \in \mathbb{C}^{N_R\times N_T} $ is a diagonal matrix, whose diagonal elements $\lambda_1\geq \lambda_2 \geq \cdots \geq \lambda_N$ are the ordered singular values of the channel matrix $\mathbf{H}$. We denote with $N$ the rank of $\mathbf{H}$, $N \coloneqq  \min(N_T,N_R)$. If we multiply \eqref{Model:MIMO} with the unitary matrix $\mathbf{U}^\h$, we obtain then:
\begin{equation}
\underbrace{\mathbf{U}^\h \bs{y}}_{\coloneqq \tilde{\bs{y}}}= \mathbf{U}^\h\mathbf{U}\mathbf{\Lambda}\underbrace{\mathbf{V}^\h \bs{x}}_{\coloneqq \tilde{\bs{x}}}+\underbrace{\mathbf{U}^\h\bs{\xi}}_{\coloneqq \tilde{\bs{\xi}}}.
\end{equation}
The RV $\tilde{\bs{\xi}}$ has the same distribution as $\bs{\xi}$, i.e., $\tilde{\bs{\xi}} \sim \mathcal{N}_{\mathbb{C}}(\bs{0}_{N_T},\sigma^2 \mathbf{I}_{N_R})$ and we have:
\begin{equation*}
\mathbb{E}[\tilde{\bs{x}}^\h\tilde{\bs{x}}]= \mathbb{E}[ \bs{x}^\h \mathbf{V} \mathbf{V}^\h \bs{x}]=  \mathbb{E}[\bs{x}^\h \bs{x}].
\end{equation*}
Thus, we obtain the $N$ independent scalar Gaussian channels depicted in Fig. \ref{Fig:SVD_Dec}.
\begin{figure*}
	\centering 
	\input{figures/SVD_Decomposition.tex}
	\caption{Decomposition of the MIMO channel into $N$ parallel channels through SVD.}
	\label{Fig:SVD_Dec}
\end{figure*}

\begin{equation}
    \tilde{y}_l=\lambda_l\tilde{x}_l + \tilde{\xi}_l,\quad l=1,\ldots,N.
\end{equation}
The SVD can be interpreted as a pre-processing (multiplication with $\mathbf{V}$) and a post-processing (multiplication with $\mathbf{U}^\h$).
The optimization problem in \eqref{eq:capacity_MIMO} is reduced to \cite{Tse}
 \begin{equation}
 \begin{split}
 C(P,N_T\times N_R) = \max_{\tilde{P}_1,\ldots,\tilde{P}_N} \sum_{l=1}^{N} \log\left(1+ \frac{\lambda_l^2}{\sigma^2} \tilde{P}_l \right), \\
 \text{s.t.}  \sum_{l=1}^{N} \tilde{P}_l = P \quad \text{and} \quad \tilde{P}_l \geq 0,\ l=1,\ldots,N. \end{split}
 \end{equation}
The power $\tilde{P}_l$ is determined as follows \cite{Tse}.
\begin{equation}
\tilde{P}_l=\max \left(0, \mu-\frac{\sigma^2}{\lambda_l^2} \right),\quad l=1,\ldots,N.
\end{equation}
The expression above to compute $\tilde{P}_l$ is called the \emph{waterfilling} rule and $\mu$ is chosen to guarantee $\sum_{l=1}^{N} \tilde{P}_l = P$. Each of the $\lambda_l$ corresponds to an eigenmode of the channel, also called eigenchannel. In Fig. \ref{waterfilling}, we consider an example with four non-zero singular values $\lambda_1,\ldots,\lambda_4$. 
If the channel gain $\frac{\lambda_l^2}{\sigma^2}$ is large enough, power $\tilde{P}_l$ is given to the respective eigenmode. Otherwise, we allocate no power to the eigenchannel.
\begin{figure}[hbt!]
	\centering 
	\input{figures/waterfilling.tex}
	\caption{Example of waterfilling with five eigenvalues.}
	\label{waterfilling}
\end{figure}
\subsection{ Identification Capacity of the Gaussian MIMO Channel} 
We denote with $C(g,\tilde{P}_l)$ and $C_{ID}(g,\tilde{P}_l)$ the Shannon capacity and the identification capacity of each Gaussian sub-channel $l$, respectively. We assume that we decompose the channel in \eqref{Model:MIMO} as described above in Fig. \ref{Fig:SVD_Dec}. We have a serial-to-parallel conversion, where we convert each signal vector $\bs{x}$ into $N$ parallel scalar signals $\tilde{x}_l,\ 1\leq l\leq N$. Each signal vector lies in a Cartesian product of input sets. We perform transmit beamforming by multiplying each signal vector $\bs{x}$ with the matrix $\mathbf{V}^\h$. $\tilde{x}_l$ results from multiplying $\bs{x}$ with the $l$-th column of $\mathbf{V}$. We choose each power $\tilde{P}_l$ by waterfilling to maximize the capacity $C(P,N_T\times N_R)$ . Since we probably lose information via processing, we have:
\begin{equation}
C_{ID}(P,N_T\times N_R) \geq \sum_{l=1}^{N} C_{ID}(g,\tilde{P}_l) ,\label{eq:processing_ID}
\end{equation}
where $C_{ID}(P,N_T\times N_R)$ denotes the identification capacity of the MIMO Gaussian channel in \eqref{Model:MIMO}. Furthermore, we know from Theorem. \ref{theorem:ID_Gaussian} that for each sub-channel $l$ the following equation holds. 
\begin{equation}
C(g,\tilde{P}_l) =C_{ID}(g,\tilde{P}_l) .\label{eq:Gaussian_ID_Shannon} 
\end{equation}
From \eqref{eq:processing_ID} and \eqref{eq:Gaussian_ID_Shannon} we deduce a lower-bound on the ID capacity of the MIMO channel:

\begin{equation}
\begin{split}
C_{ID}(P,N_T\times N_R) & \geq \sum_{l=1}^{N} C(g,\tilde{P}_l) \\
& = C(P,N_T\times N_R) .
\end{split}
\end{equation}
The upper-bound on the ID capacity of the Gaussian MIMO channel is established in \cite{ITW_paper}.
{\color{black}{
\begin{remark}
The characterization of the identification capacity and secure identification capacity for MIMO Gaussian channels in this paper shows that similar to message transmission, signal processing, and channel coding can also be performed separately. This solution allows identification and secure identification to be directly integrated as optimization tasks in resource allocation. This enables important techniques of resource allocation \cite{Resource_Alocation1,Resource_Allocation2} to be used directly. 
\end{remark}}}
 \section{Conclusions} \label{sec:conclusions}
We have examined identification over Gaussian channels for their practical relevance. We have extended the direct proof of the identification coding theorem to the Gaussian case. We recall that the transition from the discrete case to the Gaussian case is not obvious. Burnashev has shown that for the white Gaussian noise channel without bandwidth constraint and finite Shannon capacity, the corresponding identification capacity is infinite. The main focus of the paper was to calculate the secure identification capacity for the GWC. Secure identification for discrete alphabets has been extensively studied over the recent decades due to its important potential use in many future scenarios. For discrete channels, it was proved that secure identification is robust under channel uncertainty and against jamming attacks. However, for continuous alphabets, no results have yet been established. This problem seems to be difficult because the wiretapper, in contrast to the discrete case, is not limited anymore and has an infinite alphabet. This is advantageous for the wiretapper since he has no limit on the hardware resolution. This means that the received signal can be resolved with infinite accuracy. In this paper, we completely solved the Gaussian case for its practical relevance. In particular, we have provided a coding scheme for secure identification over the GWC and have determined the corresponding capacity.
 Theorem \ref{main Theorem} says that secure identification is possible for continuous channels, particularly for the Gaussian case. We also studied identification over single-user Multiple-Input Multiple-Output (MIMO) channels, which are of major interest for wireless communications. We derived a lower bound on the corresponding capacity. As a direct continuation of this work, it would be interesting to investigate message identification for continuous-time channels. Future research might also explore secure identification over MIMO channels. 

\appendix
\section{Appendix}
 
\subsection{Quantization of the Gaussian Wiretap Channel} \label{QuantizationWiretapAppendix}
 We consider the same GWC $(\Wg,\Vg)$ described in \eqref{GaussianChannel}, where each input signal $x_i$ lies in the set $[-\sqrt{nP},\sqrt{nP}]$ for all $i\in \{1,\ldots,n\}$. The output sets are $\sety$ and $\setz$, which lie in the set of real numbers. For convenience and consistency with the quantization scheme used in \cite{Bur00}, we set $P=a^2$. We modify the scheme in \cite{Bur00} to be applicable on the GWC. Note that g and g', as normal distributions, satisfy the regularity conditions cited in \cite{Burnashev}. This means, there exist some constants $K, K_1,\gamma,\alpha$ such that:\\
 For $u>0$ and $1<\gamma\leq 2$
 	\begin{align}
 \int_{-\infty}^{\infty}\left(\max_{|t-x|\leq u} \sqrt{\text{g}(t)} - \min_{|t-x|\leq u} \sqrt{\text{g}(t)} \right)^2 dx \leq Ku^\gamma \label{Const1}, \end{align}
 For $z>0$ and $\alpha>2 $
 \begin{align}
 &\int_{|x| \geq z} \text{g}(x) dx \leq K_1z^{-\alpha} \label{Const2}. \end{align}
 \begin{equation}
1/\alpha + 1/\gamma < 1. \label{Const3} 
 \end{equation}
This also holds for g', i.e., there exist some constants $K', K'_1,\gamma',\alpha'$ such that the aforementioned regularity conditions are fulfilled. It can easily be checked that we can choose constants as follows:
\begin{align*}
&K=\frac{1}{\sigma^2},\ K'=\frac{1}{\sigma'^2}, \\
&K_1=K'_1=1, \\
&\gamma=\gamma'=2, \\
&\alpha=\alpha'=3.
\end{align*}
For more details, we refer the reader to \cite{MasterThesis}. The quantization of $(\Wg,\Vg)$ consists of the two following steps.
\begin{enumerate}
\item {\bf{Quantization of the Input Set $[-\sqrt{nP},\sqrt{nP}]$}}\\
 First, we quantize the interval $[-\sqrt{nP},\sqrt{nP}]$ by selecting within it the lattice $\mathcal{L}_x$ with span $\delta_x$. The span $\delta_x$ is the maximum distance between two points in the one-dimensional lattice $\setl_x$. We approximate any input vector $x^n=(x_1,x_2,\ldots,x_n) \in \setx_{n,P}$ by $\hat{x}^n=(\hat{x}_1,\hat{x}_2,\ldots,\hat{x}_n) \in \mathcal{L}_x^n$ such that $\hat{x}_i$ is the closest to $x_i$ in terms of Euclidean distance with $|\hat{x}_i| \leq |x_i|$. We denote the output measures generated by the input vector $x^n$ on $\Wg$ and $\Vg$ by $Q_{x^n}\Wg^n(\cdot)$ and $Q_{x^n}\Vg^n(\cdot)$ respectively. Analogously, $Q_{\hat{x}^n}\Wg^n(\cdot)$ and $Q_{\hat{x}^n}\Vg^n(\cdot)$ are the generated output measures by the input $\hat{x}^n$ on $\Wg$ and $\Vg$, respectively. It follows from $(36)$ in \cite{Bur00} that for $\delta > 0$:
 	\begin{align}
 	d\left(Q_{x^n}\Wg^n(\cdot),Q_{\hat{x}^n}\Wg^n(\cdot)\right) & \leq  2 \sqrt{Kn\delta_x} \leq 2\delta, \label{discreteinput1}\\
 	d\left(Q_{x^n}\Vg^n(\cdot),Q_{\hat{x}^n}\Vg^n(\cdot)\right) & \leq  2 \sqrt{K'n\delta_x} \leq 2\delta. \label{discreteinput2}
 	\end{align}
 	We define $K_{max}$ as follows:
 	 \begin{equation} K_{max}=\max(K,K'). \end{equation} We set $\delta_x=\left(\frac{\delta^2}{nK_{max}}\right)^{\frac{1}{\gamma}}$ and denote the cardinality of $\mathcal{L}_x$ by $L_x$. We have:
 	\begin{equation}
 	L_x \leq \frac{2a}{\delta_x}+2= 2a\left(K_{max}n\delta^2 \right)^{\frac{1}{\gamma}}+2.
 	\end{equation}
 	We denote the new quantized channels with discrete input alphabet $\setl_x$ by $W_{g,\setl_x}$ and $V_{\text{g}',\mathcal{L}_x}$.
 	 Let $\pi(x^n)$ be any probability distribution on $\setx_{n,P}$ generating the output measures $Q_{\pi}^nW^n_{\text{g}}(\cdot)$ on $\sety^n$ and $Q_{\pi}^n\Vg^n(\cdot)$ on $\setz^n$. From \eqref{discreteinput1} and \eqref{discreteinput2}, we deduce that we can find some input distribution $\hat{\pi}(x^n) \in \mathcal{P}(\setl_x^n) $ generating  output measures $Q_{\hat{\pi}}^nW^n_{\text{g},\mathcal{L}_x}(\cdot)$ and $Q_{\hat{\pi}}^nV^n_{\text{g}',\mathcal{L}_x}(\cdot)$ on $W_{\text{g},\mathcal{L}_x}$ and $V_{\text{g}',\mathcal{L}_x}$, respectively such that for $\delta>0$
 	 \begin{align}
 	 d\left(Q^n_{\pi}\Wg^n(\cdot),Q^n_{\hat{\pi}}W^n_{\text{g},\mathcal{L}_x}(\cdot)\right) & \leq 2\delta, \label{input} \\
 	 d\left(Q^n_{\pi}\Vg^n(\cdot),Q^n_{\hat{\pi}}V^n_{\text{g}',\mathcal{L}_x}(\cdot)\right) & \leq 2\delta.
 	 \end{align}
 	 The new noise functions defined on $[-z_0,z_0]$ are denoted by g$_{z_0}$ and g'$_{z_0}$.
 \item {\bf Quantization of $\sety$ and $\setz$} \\
 We first consider the output set $\sety$. 
 We proceed as described in \cite{Bur00}. First, we choose $z_0= a+\left(\frac{K_1n}{\delta}\right)^{\frac{1}{\alpha}}$. We assign $z_0$ to all outputs $y$ with $|y|> z_0$ and denote the new channel with discrete input set $\setl_x$ and bounded output set $[-z_0,z_0]$ by $W_{\text{g},\mathcal{L}_x,z_0}$. Thus, if $\hat{\pi}(x^n) \in \mathcal{P}(\setl_x^n)$ generates the output measures $Q_{\hat{\pi}}^nW^n_{\text{g},\mathcal{L}_x}(\cdot)$ and $Q_{\hat{\pi}}^nW^n_{\text{g},\mathcal{L}_x,z_0}(\cdot)$ on $W_{\text{g},\mathcal{L}_x}$ and $W_{\text{g},\mathcal{L}_x,z_0}$, respectively, then for $\delta>0$
 \begin{equation}
 d\left(Q_{\hat{\pi}}^nW^n_{\text{g},\mathcal{L}_x}(\cdot),Q_{\hat{\pi}}^nW^n_{\text{g},\mathcal{L}_x,z_0}(\cdot)\right) \leq 2\delta. \label{z0}
 \end{equation}
 Now, we choose in $[-z_0,z_0]$ the lattice $\mathcal{L}_y$ with span $\epsilon$, where
 \begin{equation}
 \epsilon=\left(\frac{\delta^2}{K_{max}n}\right)^{\frac{1}{\gamma}},\quad \delta>0.
 \end{equation}
 

 We correspondingly approximate g$_{z_0}$ by a piece-wise constant noise function $\text{g}_{\epsilon}$. 
 We denote the cardinality of $\mathcal{L}_y$ by $L_y$. We have:
 \begin{equation}
 L_y \leq \frac{2z_0}{\epsilon} + 4 = 2 \left(\frac{Kn}{\delta^2} \right)^{\frac{1}{\gamma}} \left(a+ \left(\frac{nK_1}{\delta}\right)^{\frac{1}{\alpha}} \right) +4.
 \end{equation}
 
 We denote the new channel with discrete input alphabet $\setl_x$ and discrete output alphabet $\setl_y$ by $\tW$. If $\hat{\pi}(x^n) \in \mathcal{P}(\setl_x^n)$ generates the output measures $Q_{\hat{\pi}}^nW^n_{\text{g},\mathcal{L}_x,z_0}(\cdot)$ and $Q_{\hat{\pi}}^n\tW^n(\cdot)$ on $W_{\text{g},\mathcal{L}_x,z_0}$ and $\tW$, respectively, then for $\delta>0$
 \begin{equation}
 d\left( Q_{\hat{\pi}}^nW^n_{\text{g},\mathcal{L}_x,z_0}(\cdot),Q_{\hat{\pi}}^n\tW^n (\cdot) \right) \leq 2\delta. \label{output}
 \end{equation}
 From \eqref{input}, \eqref{z0} and \eqref{output}, we deduce that if $\pi(x^n)$ is a probability distribution on $\setx_{n,P}$ generating the output measures $Q_{\pi}^n\Wg^n(\cdot)$ on $\Wg$, we can find some input distribution $\hat{\pi}(x^n),\quad \hat{\pi}(x^n) \in \mathcal{P}(\setl_x^n) $ generating the output measure $Q_{\hat{\pi}}^n\tW^n(\cdot)$ on $\tW$ such that for $\delta>0$
 \begin{align}
 d\left(Q_{\pi}^n\Wg^n(\cdot),Q_{\hat{\pi}}^n\tW^n(\cdot)\right) \leq 6\delta,\quad \delta>0. 
 \end{align}
 The quantization of $\setz$ follows the same scheme because $\sety=\setz=\mathbb{R}$ and the quantization of $\sety$ and $\setz$ depends on $K_{max}$ and the constants $K_1,\alpha, \gamma$. These constants are the same for both noise functions $\text{g}$ and $\text{g}'$. We have 
 \begin{equation}
     d\left(Q_{x^n}\Vg^n(\cdot),Q_{\hat{\pi}}^n\tV^n(\cdot)\right) \leq 6\delta,\quad \delta>0,
 \end{equation}
 where $\setl_z$ is the new discrete output set of Eve's channel and $\tV$ denotes the new discrete Eve's channel. $Q_{x^n}\Vg^n(\cdot)$ is the output measure generated by the input $Q_{x^n}$ on $\Vg$ and $Q_{\hat{\pi}}^n\tV^n(\cdot)$ is the output measure generated by the input $Q_{\hat{\pi}}^n$ on $\tV$.
\begin{remark}
If the wiretap channel has a positive secrecy capacity $C(\text{g},\text{g}',P)$, $K_{max}=\max(\frac{1}{\sigma^2},\frac{1}{\sigma'^2})=\frac{1}{\sigma^2}=K$.
\end{remark}
\end{enumerate}
  \subsection{Quantization of a Discrete-Time Channel with Additive Noise and Infinite Input and Output Alphabets} 
  
  \label{appendixB}
 Burnashev considered in \cite{Bur00} the approximation of output measures for channels with infinite alphabets. In this section, we give a summary of the quantization scheme, which we modified and applied on the GWC.
 \subsubsection{Channel Model} 
 The following additive noise channel is considered.
 \begin{equation}
     y_i=x_i+\xi_i,\quad \forall i \in \{1,\ldots,n\}.
 \end{equation}
 $x^n=(x_1,x_2,\ldots,x_n)$ is the channel input sequence and $y^n$ is the output sequence.
 	$\xi_i$ are iid for all $i \in \{1,\ldots,n\}$ and drawn from a noise function f satisfying the regularity conditions described in \cite{Burnashev}. 
 The channel input fulfills the following peak power constraint:  \begin{equation}|x_i| \leq  a,\quad a>0,\quad \forall i=1,\ldots,n .
 \label{peakConstraint}
 \end{equation}
 It follows from \eqref{peakConstraint} that the new constrained input set is $\setx=[-a,a]$.
 The output set is infinite $\sety= \mathbb{R}$.
 The Shannon capacity of the channel $W_{\text{f}}$ is denoted by $C(\text{f},a)$.
 
\subsubsection{Quantization Steps}

The channel $W_{\text{f}}$ should be approximated by a discrete channel with input and output alphabets $\setl_x$ and $\setl_y$, respectively. $\setl_x$ and $\setl_y$ should fulfill the following inequalities:
\begin{align}
\lim_{n \to \infty} \frac{|\setl_x|}{n} \ln n & = 0, \\
\lim_{n \to \infty} \frac{|\setl_y|}{n} \ln^4 n & =0.
\end{align}
 \begin{enumerate}
 	\item 
 	 $\setx=[-a,a]$ is quantized by choosing in it a lattice $\setl_x$ with span $\delta_x$. Any input vector $u^n \in \setx^n$ is approximated by $u^{\prime n} \in \setl_x^n$. $u^{\prime}_i$ is chosen to be the closest to $u_i$ in terms of euclidean distance with $|u_i^{\prime}| \leq |u_i|$. Let $Q_{u^n}W_{\text{f}}^n(\cdot)$ and $Q_{u^{\prime n}}W_{\text{f}}^n(\cdot)$
 	 be the output measures generated by $u^n \in \setx^n$ and $ u^{\prime n} \in \setl_x^n$ over $W_{\text{f}}$, respectively. We have
 	\begin{align}
 	d\left(Q_{u^n}W_{\text{f}}^n(\cdot),Q_{u^{\prime n}}W_{\text{f}}^n(\cdot)\right) & \leq 2\left(K \sum_{i=1}^{n} |u_i-u^{\prime}_i|^\gamma\right)^{\frac{1}{2}} \nonumber\\
 	& \leq 2 \sqrt{Kn\delta_x^\gamma} .
 	\end{align}
 	We set $\delta_x=\left( \frac{\delta^2}{nK} \right)^{\frac{1}{\gamma}}$. The cardinality of $\setl_x$ is denoted by $L_x$ and is upper-bounded as follows:
 	\begin{equation}
 	L_x \leq \frac{2a}{\delta_x} +2 .
 \end{equation}
 The original channel $W_{\text{f}}$ is then approximated by the channel ${W}_{\text{f},\setl_x}$ with finite input alphabet $\setl_x$.
 Generally, let $\pi(x^n)$ be any prior distribution on $\setx^n$ generating $Q^n_{\pi}W^n_{\text{f}}(\cdot)$ on $\sety^n$. $\pi^\prime(x^n),\ x^n \in \setl_x^n$ is chosen such that the generated output measure $Q^n_{\pi^\prime}{W}^n_{\text{f},\setl_x}(\cdot)$ satisfies the following inequality:
 \begin{equation}
 d\left(Q^n_{\pi}W^n_{\text{f}}(\cdot),Q^n_{\pi^\prime}{W}^n_{\text{f},\setl_x}(\cdot)\right) \leq 2\delta,\quad \delta>0. \label{firstApp}
 \end{equation}
 \item Now, the output set $\sety$ should be bounded. For this purpose, the following quantization function is defined.
 \begin{align} 
 & l \colon  \mathbb{R} \longrightarrow [-z_0,z_0], \\
 & l(y)= \begin{cases}
 y & y \in [-z_0,z_0], \\ z_0 & \text{else},
 \end{cases}
 \end{align}
 where $z_0=a+b,\quad b>0$. The probability that an output vector $y^n$ is not in $[-z_0,z_0]^n$ can be computed as follows.
 \begin{align}
 \Pr\{y^n \notin [-z_0,z_0]^n \} & \overset{(a)}{\leq} \sum_{i=1}^{n} \Pr\{y_i \notin [-z_0,z_0] \} \\
 & \overset{(b)}{=} \sum_{i=1}^{n} \int_{|\xi_i| \geq z_0-x_i} \text{f}(\xi_i) d\xi_i \\ & \overset{(c)}{\leq} \sum_{i=1}^{n} K_1 (z_0-x_i)^{-\alpha}  \\ & \overset{(d)}{\leq} K_1 n b^{-\alpha} .
 \end{align}
 $(a)$ follows by the union bound. $(b)$ follows because $y_i=x_i+\xi_i$. $(c)$ follows by \eqref{Const2}. $(d)$ follows because $x_i \leq a,\quad \forall i=1,\ldots,n$. \\
If we choose $b=\left(\frac{K_1n}{\delta}\right)^{\frac{1}{\alpha}}$, we have then
\begin{equation}
\Pr\{y^n \notin [-z_0,z_0]^n \} \leq \delta, \quad \delta>0.
\end{equation} 
Thus, the channel $W_{\text{f},\setl_x}$ is approximated by another channel $W_{\text{f},\setl_x,z_0}$ with bounded output alphabet. Let $\pi^\prime(x^n),\ x^n \in \setl_x^n$ be any prior distribution generating $Q^n_{\pi^\prime}{W}^n_{\text{f},\setl_x}(\cdot)$ and $Q^n_{\pi^\prime}{W}^n_{\text{f},\setl_x,z_0}(\cdot)$ on $W_{\text{f},\setl_x}$ and $W_{\text{f},\setl_x,z_0}$, respectively. We have 
\begin{equation}
d\left(Q^n_{\pi^\prime}{W}^n_{\text{f},\setl_x}(\cdot),Q^n_{\pi^\prime}{W}^n_{\text{f},\setl_x,z_0}(\cdot)\right) \leq 2\delta. \label{secondApp}
\end{equation}
\item Now $[-z_0,z_0]$ should be approximated by a finite set.
 $$[-z_0,z_0]=\cup_{i=-i_0}^{i_0} [i\epsilon,(i+1)\epsilon)+ \{z_0\},\quad i_0=\frac{z_0}{\epsilon},\ \epsilon>0.$$  $[-z_0,z_0]$ is quantized by choosing in it the lattice $\setl_y$.
$$ \setl_y= \{y \colon y=i\epsilon,\quad i=-i_0,\ldots,i_0\},\quad i_0=\frac{z_0}{\epsilon}.$$
Next, the noise function f is approximated by a piece-wise constant density $\text{f}_{\epsilon}$ defined as follows:
\begin{equation*}
\text{f}_\epsilon(y-x)= \begin{cases} \frac{1}{\epsilon} \int_{i\epsilon}^{(i+1)\epsilon} \text{f}(t-x) dt, & y \in \setd_\epsilon \\
\text{f}(z_0-x) + \int_{|\xi|>z_0} \text{f}(\xi-x) d\xi,  & y=z_0, \end{cases}
\end{equation*}
where $\setd_\epsilon= [i\epsilon, (i+1)\epsilon)$. The new discrete channel is denoted by $\tilde{W}_{\text{f}}$, which has finite input and output alphabets $\setl_x$ and $\setl_y$, respectively. If $\pi^\prime(x^n),\ x^n \in \setl_x^n$ is any prior distribution generating $Q^n_{\pi^\prime}W^n_{\text{f},\setl_x,z_0}(\cdot)$ and $Q^n\tilde{W}^n_{\text{f}}(\cdot)$ on $W_{\text{f},\setl_x,z_0}$ and $\tilde{W}_{\text{f}}$, respectively, then for $\delta>0$:
\begin{equation}
d\left(Q^n_{\pi^\prime}{W}^n_{\text{f},\setl_x,z_0}(\cdot),Q^n\tilde{W}^n_{\text{f}}(\cdot)\right) \leq 2\delta. \label{thirdApp}
\end{equation}
\end{enumerate}
Combining \eqref{firstApp}, \eqref{secondApp} and \eqref{thirdApp}, we deduce the following lemma \cite{Bur00}. 
\begin{lemma}\cite{Bur00}
If $\pi(x^n)$ any prior distribution on $\setx^n$ generating $Q^n_{\pi}W^n_{\text{f}}(\cdot)$ on $\sety^n$, then we choose some prior distribution $\pi^\prime(x^n),\ x^n \in \setl_x^n$ generating on the channel $\tilde{W}_{\text{f}}$ an output measure $Q^n\tilde{W}^n_{\text{f}}(\cdot)$ such that the following inequality holds:
\begin{equation}
d\left(Q^n_{\pi}W^n_{\text{f}}(\cdot),Q^n\tilde{W}^n_{\text{f}}(\cdot)\right) \leq 6\delta,\quad \delta>0.
\end{equation}
\end{lemma} 
 
\begin{remark}
Burnashev extended this quantization to additive noise channels with average power constraint, i.e., the input $x^n=(x_1,x_2,\ldots,x_n)$ satisfies the following inequality:
\begin{equation}
\frac{1}{n} \sum_{i=1}^{n} x_i^2 \leq P.
\end{equation}
Thus, the same quantization scheme can also be applied to Gaussian channels with an average power constraint.
 For a more detailed description, we refer the reader to \cite{Bur00}.
\end{remark}

\section*{Acknowledgments}

The authors acknowledge the financial support by the Federal Ministry of Education and Research of Germany in the program of “Souverän. Digital. Vernetzt.”. Joint project 6G-life, project identification number: 16KISK002.{\color{black}{ Holger Boche and Christian Deppe further gratefully acknowledge
the financial support by the BMBF Quantum Programm QD-CamNetz, Grant
16KISQ077, QuaPhySI, Grant 16KIS1598K, and QUIET, Grant 16KISQ093.}} Christian Deppe and Wafa Labidi were supported by the Bundesministerium 
f\"ur Bildung und Forschung (BMBF) through Grants 16KIS1005 and 16KIS1003K, {\color{black}{respectively}}.
%




%
\bibliographystyle{IEEEtran}
\bibliography{IEEEabrv,confs-jrnls,refrences}

\end{document}

%% file: figures/Gaussian_Channel.tex
\tikzstyle{farbverlauf} = [ top color=white, bottom color=white!80!gray]
\tikzstyle{block} = [draw,top color=white, bottom color=white!80!gray, rectangle, rounded corners,
minimum height=2em, minimum width=2.5em]
\tikzstyle{input} = [coordinate]
\tikzstyle{sum} = [draw, circle,inner sep=0pt, minimum size=2mm,  thick]
\tikzstyle{arrow}=[draw,->] 
\begin{tikzpicture}[auto, node distance=2cm,>=latex']
\node[] (M) {$L$};
\node[block,right=.5cm of M] (enc) {Encoder};
\node[sum, right=1cm of enc] (channel) {$+$};
\node[block, right=1cm of channel] (dec) {Decoder};
\node[right=.5cm of dec] (Mhat) {$\hat{L}$};
\node[above=1cm of channel] (noise) {$\xi^n$};
\draw[->] (M) -- (enc);
\draw[->] (enc) --node[above]{$x^n$} (channel);
\draw[->] (noise) -- (channel);
\draw[->] (channel) --node[above]{$y^n$} (dec);
\draw[->] (dec) -- (Mhat);
\end{tikzpicture}

%% file: figures/wiretap.tex
\scalebox{.8}{\tikzstyle{block} = [draw, top color=white, bottom color=white!80!gray, rectangle, rounded corners,
minimum height=2em, minimum width=2cm]
\tikzstyle{blockchannel} = [draw, top color=white, bottom color=white!80!gray, rectangle, rounded corners,
minimum height=2cm, minimum width=.1cm]
\tikzstyle{input} = [coordinate]
\usetikzlibrary{arrows}
\begin{tikzpicture}[scale= 1,font=\footnotesize]
\node[] (m) {\small $L$};
\node[block,right=.5cm of m] (enc) {\small Encoder};
\node[blockchannel, right=.7cm of enc](channel) {\small Channel
$(W,V)$};
\node[block,right= .7cm of channel.390] (bob) {\small Decoder};
\node[block,right=.7cm of channel.330] (eve) {\small Eavesdropper};
\node[right=.5cm of bob] (what) {\small $\hat{L}$};
\node[draw,circle,minimum size=.5cm,inner sep=0pt, right=.5cm of eve] (wbar) {\small $\cancel{L}$ }; 

\draw[->] (m) -- (enc);
\draw[->] (enc) -- node[above]{$X$}  (channel);
\draw[->] (channel.390) -- node[above]{$Y$} (bob);
\draw[->] (channel.330) -- node[above]{$Z$} (eve);
\draw[->] (bob) -- (what);
\draw[->] (eve) -- (wbar);
\end{tikzpicture}}

%% file: figures/CodeConstrWiretap.tex
\scalebox{.8}{\begin{tikzpicture}
	\definecolor{mygreen}{RGB}{80,160,80}
	\definecolor{indiagreen}{rgb}{0.07, 0.53, 0.03}
\tikzset{pblock/.style = {rectangle split, rectangle split horizontal,
 rectangle split parts=2,draw,inner sep=1ex , align=center}}
\node[pblock] (C) at (2,-4){\nodepart[text width=1cm]{one} $\up$
\nodepart{two}$\upp$};
\draw [decorate,decoration={brace,amplitude=6pt},xshift=-4pt,yshift=0pt]
 ($(C.two north)+(-0.4,-.7)$)--($(C.one north)+(-.65,-.7)$) node [black,midway,below,yshift=-.2cm]
{\footnotesize $n$};
\draw [decorate,decoration={brace,amplitude=6pt},xshift=-4pt,yshift=0pt]
 ($(C.two)+(0.55,-.2)$) -- ($(C.one)+(1.15,-.2)$) node [black,midway,below,yshift=-.1cm]
{\footnotesize $\lceil \sqrt{n}\rceil$};
\draw [decorate,decoration={brace,amplitude=6pt},xshift=-4pt,yshift=0pt]
 ($(C.one north)+(-.7,0)$)--($(C.two north)+(.4,0)$) node [black,midway,above,yshift=.1cm]
{\footnotesize $m$};
\node[rectangle split, rectangle split parts=4,draw,inner sep=1ex] (A) at (.5,-1.5)
{$\up_1$\nodepart{two}$\up_2$ \nodepart{three}$\vdots$\nodepart{four}$\up_{\Mp}$ }
;
\node[above=.1cm of A] (A1) {\small$\Cp=\{(\up_j,\dcp_j),j \in \{1,\ldots,\Mp\} \}$};
\node[above=.7cm of A] (A1) {transmission code};
\draw [decorate,decoration={brace,amplitude=10pt},xshift=-4pt,yshift=0pt]
(A.south west)-- (A.north west) node [black,midway,xshift=-1.5cm]
{\footnotesize $|\Cp|=\lceil2^{n(C-\epsilon)}\rceil$};
\node[rectangle split, rectangle split parts=4,draw,inner sep=.5ex] (B) at (5.5,-1.5)
{$\upp_1$\nodepart{two}$\upp_2$ \nodepart{three}$\small \vdots$\nodepart{four}$\upp_{\Mpp}$} ;
\node[above=.4cm of B] (B1) {\small$\Cpp=\{(\upp_k,\dcpp_k),k \in \{1,\ldots,\Mpp\} \}$};
\node[above=1cm of B] (B1) {{ wiretap code} \includegraphics[width=.06\linewidth]{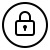}};
\draw [decorate,decoration={brace,amplitude=10pt},xshift=-4pt,yshift=0pt]
(B.north east)-- (B.south east) node [black,midway,xshift=1.5cm]
{\footnotesize $|\Cpp|=\lceil2^{\sqrt{n}\epsilon}\rceil$};
\node[right=.1cm of C] (C1) {$\in \C=\{(Q_i,\dc_i), i \in \left\{1,\ldots,N\right\} \}$ };
\draw [white]
($(C.two)+(0.55,0.7)$) -- ($(C.one)+(1.15,0.7)$) node [white,midway,below,yshift=.8cm]
{\includegraphics[width=.06\linewidth]{icons8-secure-50.png}};
\draw[->] (A.three east) -- (C.one north);
\draw[->] (B.three west) -- (C.two north);
\end{tikzpicture}}

%% file: figures/MIMOModel.tex
\tikzstyle{farbverlauf} = [ top color=white, bottom color=white!80!gray]
\tikzstyle{vectorarrow} = [thick, decoration={markings,mark=at position
   1 with {\arrow[semithick]{open triangle 60}}},
   double distance=1.4pt, shorten >= 5.5pt,
   preaction = {decorate},
   postaction = {draw,line width=1.4pt, white,shorten >= 4.5pt}]
\tikzstyle{block} = [draw,top color=white, bottom color=white!80!gray, rectangle, rounded corners,
minimum height=2em, minimum width=2.5em]
\tikzstyle{input} = [coordinate]
\tikzstyle{sum} = [draw, circle,inner sep=0pt, minimum size=2mm,  thick]
\tikzstyle{arrow}=[draw,->] 
\begin{tikzpicture}[auto, node distance=2cm,>=latex']
\node[] (x) {\large$\boldsymbol{x}$};
\node[block,right=1cm of x] (channel) {\large$\mathbf{H}$};
\node[sum, right=1cm of channel] (sum) {$+$};
\node[right=1cm of sum] (y) {\large$\boldsymbol{y}$};
\node[above=.5cm of sum] (noise) {$\boldsymbol{\xi}$};
\draw[vectorarrow] (x) -- (channel);
\draw[vectorarrow] (channel) -- (sum);
\draw[vectorarrow] (noise) -- (sum);
\draw[vectorarrow] (sum) -- (y);
\end{tikzpicture}

%% file: figures/SVD_Decomposition.tex
\tikzstyle{vecArrow} = [thick, decoration={markings,mark=at position
   1 with {\arrow[semithick]{open triangle 60}}},
   double distance=1.4pt, shorten >= 5.5pt,
   preaction = {decorate},
   postaction = {draw,line width=1.4pt, white,shorten >= 4.5pt}]
\tikzstyle{block} = [draw, top color=white, bottom color=white!80!gray, rectangle, rounded corners,
minimum height=2em, minimum width=2cm]
\tikzstyle{blockSVD} = [draw, top color=white, bottom color=white!80!gray, rectangle, rounded corners,
minimum height=3cm, minimum width=.9cm]
\tikzstyle{input} = [coordinate]
\tikzstyle{sum} = [draw, circle,inner sep=1pt, minimum size=2mm, thick]
\usetikzlibrary{arrows}
\begin{tikzpicture}[scale= 1,font=\footnotesize]
\node[] (xtilde) {\large$\boldsymbol{\tilde{x}}$};
\node[blockSVD,right=1cm of xtilde] (prep) {\large$\mathbf{V}$};
\node[blockSVD, right=1.5cm of prep](prepp) {\large
$\mathbf{V}^\h$};
\node[sum,right= .8cm of prepp.425] (lambda1) {$\times$};
\node[sum,right= .8cm of prepp.295] (lambda2) {$\times$};
\node[blockSVD, right=2cm of prepp] (postp) {\large$\mathbf{U}$};
\node[blockSVD, right=2.8cm of postp] (post) {\large$\mathbf{U}^\h$};
\node[right=1cm of post] (ytilde) {\large$ \boldsymbol{\tilde{y}}$};
\node[above=.3cm of lambda1] (lamda1) {\large$\lambda_1$};
\node[above=.3cm of lambda2] (lamda2) {\large$\lambda_N$};
\node[sum,right=1cm of postp] (ksi) {$+$};
\node[above=.3cm of ksi] (ksi1){\large$ \boldsymbol{\tilde{\xi}}$};
\draw[vecArrow] (xtilde) -- (prep) ;
\draw[vecArrow] (prep) --node[above]{\large$\boldsymbol{x}$} (prepp);
\draw[->,thick] (prepp.425)--(lambda1);
\draw[->,thick] (prepp.295)--(lambda2);
\draw[->,thick] (lambda1)--(postp.115);
\draw[->,thick] (lambda2)--(postp.245);
\draw[vecArrow] (postp) -- (ksi);
\draw[vecArrow] (ksi) --node[above]{\large$\boldsymbol{y}$} (post);
\draw[vecArrow] (post) -- (ytilde);
\draw[->,thick] (lamda1) -- (lambda1);
\draw[->,thick] (lamda2) -- (lambda2);
\draw[->,thick] (ksi1) -- (ksi);
\draw[dashed] (3.3,-2.3) rectangle (9.3,2.3);
\node[below=.5 of prep] (precoder) { pre-processing};
\node[below=.5 of post] (decoder) { post-processing};
\node[below=0.01 cm of lambda1] (inf1) {$ \vdots$};
\node[] at (6, -2.5) {channel};

\end{tikzpicture}

%% file: figures/waterfilling.tex
\scalebox{.88}{
\pgfplotsset{
 width=\columnwidth,
  compat=newest,
  xlabel near ticks,
  ylabel near ticks
}
  \begin{tikzpicture}[font=\small]
    
    \begin{axis}[
     ymin=0,
     xmin=0.5,xmax=5.75, ymax=80,
     ybar=5pt, 
        x=1.5cm, 
        bar width=1.5cm, 
      ytick ={10, 30, 35, 40, 50},
      yticklabels={$\sigma^2/\lambda_1^2$, $\sigma^2/\lambda_2^2$ ,$\sigma^2/\lambda_3^2$ , $\mu$, $\sigma^2/\lambda_4^2$},
      xmajorgrids=false,
      axis x line=bottom,
      axis y line=left,
    xtick={1,2,3,4,5,6},
    ylabel = {power},
        xlabel = {channel index}]
        \addplot[fill=white] plot coordinates
            {(1,10) (2,30) (3,35) (4,50) (5,100)};
       \draw[dashed,-] (0.5,40) -- (5.5,40);
       \draw[<->,very thick] (1,10) -- node[left]{\tiny{$\tilde{P}_1$}} (1,40);
       \draw[<->,very thick] (2.3,30) -- node[left]{\tiny$\tilde{P}_2$} (2.3,40);
       \draw[<->,very thick] (3,35) -- node[left]{\tiny$\tilde{P}_3$} (3,40);
        \fill [blue,opacity=0.1, draw=none] (0.5,10) rectangle (1.5,40);
        \fill [blue,opacity=0.1, draw=none] (1.5,30) rectangle (2.5,40);
        \fill [blue,opacity=0.1, draw=none] (2.5,35) rectangle (3.5,40);
        \draw[-] (0.8,80) -- (0.8,50) -- (3.3,50)--(3.3,80);
        \draw[-,dashed] (0.8,75) -- (3.3,75);   
        \fill [blue,opacity=0.1, draw=none] (0.8,50) rectangle (3.3,75);
        \node[] at (2.5,60) (A) {};
        \node[] at (1.3,35) (B) {};
        \draw [->,thick] (A) to [out=200,in=50] (B);
        \node[] at (2.4,65) {$P$};
        \node[] at (4, 45) {\tiny$ \tilde{P}_4=0$};
        \node[] at (5,55) {\tiny $\tilde{P}_5=0$};
        \node[] at (5,70) { $\vdots$};
        \end{axis}
  \end{tikzpicture}}